\documentclass{article}
\usepackage[utf8]{inputenc}
\usepackage[margin=1in]{geometry}
\usepackage{cancel}
\usepackage{caption}
\usepackage{subcaption}
\usepackage{amsthm}
\usepackage{amsmath}
\usepackage{textcomp}
\usepackage{hyperref}
\hypersetup{colorlinks,allcolors=blue}
\usepackage{amsfonts}
\usepackage{mathtools}
\usepackage{amssymb}

\usepackage{fancyhdr}
\usepackage{enumerate}
\usepackage{bm}
\usepackage{thmtools,thm-restate}
\usepackage[pdftex]{graphicx}
\usepackage{comment}
\usepackage{xspace}
\usepackage{xcolor}
\usepackage[numbers]{natbib}
\newtheorem{theorem}{Theorem}%[section]

\newtheorem{lemma}[theorem]{Lemma}
\newtheorem{definition}{Definition}
\newtheorem{remark}[theorem]{Remark}

\newcommand{\Z}{\mathbb{Z}}
\newcommand{\N}{\mathbb{N}}
\newcommand{\eps}{\varepsilon}

\newcommand{\cut}{Max-Cut\xspace}
\newcommand{\lin}{Max-2Lin($q$)\xspace}

\title{Inapproximability for Local Correlation Clustering \\ and Dissimilarity Hierarchical Clustering}
\author{Vaggos Chatziafratis\thanks{Google Research NY, vaggos@google.com and vaggos@cs.stanford.edu}\ \ \ \ \ \  Neha Gupta\thanks{Stanford University, nehagupta@cs.stanford.edu}\ \ \ \ \ \ Euiwoong Lee\thanks{University of Michigan, euiwoong@umich.edu}}

\usepackage{natbib}
\usepackage{graphicx}

\begin{document}

\maketitle

\begin{abstract}
    We present hardness of approximation results for Correlation Clustering with local objectives and for Hierarchical Clustering with dissimilarity information. For the former, we study the local objective of Puleo and Milenkovic (ICML '16) that prioritizes reducing the disagreements at data points that are worst off and for the latter we study the maximization version of Dasgupta's cost function (STOC '16). Our $\mathbf{APX}$-hardness results imply that the two problems are hard to approximate within a constant of $\tfrac43\approx 1.33$ (assuming $\mathbf{P}\neq \mathbf{NP}$) and $\tfrac{9159}{9189}\approx 0.9967$ (assuming the Unique Games Conjecture)   respectively.
\end{abstract}
\section{Introduction}

Partitioning items based on pairwise similarity or dissimilarity information has been a ubiquitous task in machine learning and data mining with many different variants across sciences, depending on the form of the provided data and the desired output. For example, one of the earliest formulations in clustering is Lloyd's $k$-means objective~\cite{lloyd1982least}\footnote{This work was already written back in 1957 as a Bell labs report, but a formal publication took place in 1982.}, which was a major step towards precise ways of evaluating candidate solutions, similar to $k$-median or $k$-center and other graph $k$-cut objectives, where $k$ denotes the number of clusters in the final partition. A disadvantage shared by such $k$-partitioning formulations of clustering is the parameter $k$ itself: as $k$ may be unknown or dynamically changing with the collected data, we would like to alleviate the requirement of specifying a fixed number of clusters a priori and design non-parametric alternatives. Two established and well-studied such approaches are \textit{Correlation Clustering} and \textit{Hierarchical Clustering}.

\subsection{Correlation Clustering with Local Objectives}

In Correlation Clustering (CC)~\cite{bansal2004correlation}, we are given a graph $G$ on $n$ nodes, whose edges are labelled as $``+$'' or $``-$'' representing whether two items are similar or dissimilar respectively. In the original formulation, the goal is to produce a clustering respecting the edge labeling as much as possible, i.e., positive
edges should lie in the same cluster, whereas negative edges should be placed across different clusters. We say an edge is in \textit{disagreement} or is  \textit{misclassified} or is an \textit{error}, if it is positive yet placed across clusters, or if it is negative yet placed within a cluster. Notice that if in the optimum solution, no edge is in disagreement then the problem is trivial: simply output the connected components formed by the endpoints of $``+$" edges. Generally, no such perfect clustering exists so we want to find a partition that approximates\footnote{All approximation factors stated in this paper are multiplicative with respect to an optimum solution.} the optimum. An advantage of this formulation is that the
number of clusters is not predefined (bypassing the need to specify $k$) and naturally it has been extremely useful both in theory~\cite{bansal2004correlation,swamy2004correlation,charikar2005clustering,ailon2008aggregating,chawla2015near}) and in practice, e.g., in spam filtering~\cite{bonchi2014correlation,ramachandran2007filtering}, image segmentation~\cite{kim2011higher} and co-reference resolution~\cite{cohen2001learning,cohen2002learning, elmagarmid2006duplicate}.

The objective in CC is to find a partition $P$ that minimizes the total number of errors, however recent studies focus on a broader class of objectives sometimes referred to as the \emph{local} version of CC~\cite{puleo2016correlation} or CC with local guarantees~\cite{charikar2017local,kalhan2019correlation}.  Local-CC aims to bound the number of errors (misclassified edges) incident on any node, in other words, aims to reduce errors at nodes that are worst-off in the partition, and as such it appears in the context of fairness in machine learning~\cite{ahmadi2020fair}, bioinformatics~\cite{cheng2000biclustering}, community detection
without antagonists (nodes largely different than their
community), social sciences, recommender systems and more~\cite{kriegel2009clustering,symeonidis2006nearest}.

More specifically, given a partition $P$, let its \emph{disagreements vector}  be the $n$-dimensional vector indexed by
the nodes whose $i$-th coordinate equals the number of disagreements (induced by $P$) at node $i$. Local-CC asks to minimize the $\ell_q$-norm ($q\ge 1$) of the disagreements vector. Observe that classical CC~\cite{bansal2004correlation}  then simply corresponds to $\ell_1$-minimization. 
In terms of algorithmic results for the $\ell_q$ ($q\ge 1$) version, when the graph is complete, a polynomial time $48$-approximation was proposed in~\cite{puleo2016correlation} and was later improved to a factor $7$-approximation by Charikar et al.~\cite{charikar2017local}, who also gave
an $O(\sqrt n)$-approximation for $\ell_{\infty}$-minimization on arbitrary graphs (known as Min Max CC). The current best for complete graphs is a $5$-approximation and for general graphs an $O(n^{\tfrac12-\tfrac{1}{2q}}\log^{\tfrac12+\tfrac{1}{2q}} n)$-approximation (for $1 \le q < \infty$) by~\cite{kalhan2019correlation}. Here we complement positive approximation results by giving the first $\mathbf{APX}$-hardness for the $\ell_\infty$ version of local-CC:

\begin{theorem}
It is $\mathbf{NP}$-hard to approximate the $\ell_\infty$ version of Local Correlation Clustering within a factor of $\tfrac43$ (even on complete graphs).
\end{theorem}

\subsection{Hierarchical Clustering} 

%hierarchical clustering what is it? applications and history in biology, hyperbolic embeddings

%dasgupta addressed, definition of objective, and maximization version here as for other objectives we already know hardness, 

%approximation algorithms list what is known for the costs

Hierarchical Clustering (HC) is another fundamental problem in data analysis that does not require knowing the number $k$ of desired clusters in advance as it generates a hierarchy of clusters. Given $n$ items with their pairwise dissimilarities or similarities, the output of HC is a rooted tree $T$ with $n$ leaves that are in one-to-one correspondence with the items and with its internal nodes capturing intermediate groupings. Notice that in HC, all items form initially a cluster at the root, and successively smaller and smaller clusters are formed at internal nodes as we move towards the leaves, which should be thought of as singleton clusters. The goal in HC is to respect the given pairwise relationships as much as possible, e.g., by separating dissimilar items in the beginning (near the root) and maintaining similar items together for as much as possible (separating them close to the leaves). HC arises in various applications as data usually exhibits hierarchical structure. It originated in biology and phylogenetics~\cite{sneath1973numerical,felsenstein2004inferring} and since then can be found in cancer gene sequencing~\cite{sorlie2001gene, sotiriou2003breast}, text/image analysis~\cite{steinbach2000}, community detection~\cite{leskovec2019mining} and more. 

Most work on HC has been traditionally focused on proposing heuristics for HC (e.g., single or average linkage and other bottom up agglomerative processes like Ward's method~\cite{ward1963hierarchical}) and despite its importance, a formal understanding of HC was hindered by lack of well-posed objectives analogous to $k$-means for standard clustering. In a first attempt towards evaluating candidate hierarchical trees, Dasgupta and Long~\cite{dasgupta2005performance} proposed to compare $k$-partitions obtained by pruning the tree against the optimal $k$-center placement, for multiple values of $k$; extensions to $k$-median and relations to so-called incremental clusterings are shown in~\cite{charikar2004incremental,lin2010general}. Moreover, a recent analysis of Ward's method shows that it finds good clusterings for all levels of granularity that contain a meaningful
decomposition~\cite{grosswendt2019analysis}.  Despite those efforts, no “global” objective function was associated with the final tree output. To help address this, Dasgupta~\cite{dasgupta2016cost} introduced a discrete cost function over the space of  trees with $n$ leaves and showed that low-cost trees correspond to good hierarchical partitions in the data. Overall, Dasgupta's work ignited an objective-oriented perspective to HC with several approximation results shedding light to old algorithms like average linkage~\cite{moseley2017approximation} and designing new algorithms based on tools like Sparsest/Balanced-Cut~\cite{charikar2017approximate} (and their extensions for incorporating ancestry constraints in the tree~\cite{chatziafratis2018hierarchical}), semidefinite programs~\cite{charikar2019hierarchical} and random projections~\cite{charikar2019euclidean}, Densest-Cut~\cite{cohen2019hierarchical} and Max Uncut Bisection~\cite{ahmadian2019bisect,alon2020hierarchical} (see~\cite{chatziafratis2020hierarchical} for a survey). Finally, his formulation has recently been fruitful in continuous optimization over euclidean and hyperbolic spaces, where his 
objective is used to inform gradient-descent towards accurate embeddings of the leaves~\cite{HypHC,monath2019gradient,monath2017gradient,chierchia2019ultrametric}.

Specifically, given a weighted graph $G = (V, E, w)$, where $w\ge 0$ denotes similarity (the larger the weight the larger the similarity) Dasgupta phrased HC as the following cost minimization problem:
\begin{equation}
   \min_{\mbox{trees } T} \sum_{(i, j) \in E} w_{i, j} |T_{i,j}|\label{eq:obj}
\end{equation}
where $T_{i,j}$ denotes the subtree rooted at the lowest common ancestor of $i,j$ in $T$ and $|T_{i,j}|$ simply denotes the number of \emph{leaves} that belong to $T_{i,j}$. Notice that this objective intuitively captures the goal of HC which is to maintain similar items together for as much as possible, since separating them close to the leaves leads to lower values of this cost function. Indeed, Dasgupta showed that whenever there is a planted ground-truth clustering (e.g., stochastic block models), the tree that minimizes (\ref{eq:obj}) will recover it. Moreover, Cohen-Addad et al.~\cite{cohen2017hierarchical} proved analogous recovery results for suitably defined \emph{hierarchical} stochastic block models and Roy and Pokutta~\cite{roy2017hierarchical} used experimentally the newly-proposed cost function to obtain clusterings that correspond better to the underlying ground truth compared to those found by linkage methods. 

Perhaps not surprisingly, optimizing objective (\ref{eq:obj}) is a difficult task so the focus becomes to understand the approximability of the problem. First of all, it is known to be an $\mathbf{NP}$-hard problem~\cite{dasgupta2016cost} and actually no constant factor approximation is possible in polynomial time under certain complexity assumptions~\cite{charikar2017approximate,roy2017hierarchical}. For the complement to Dasgupta's cost studied in~\cite{moseley2017approximation}, an $\mathbf{APX}$-hardness result is provided in~\cite{ahmadian2019bisect}.

\paragraph{Dissimilarity Hierarchical Clustering} Here we focus on the \emph{dissimilarity} HC objective by Cohen-Addad et al.~\cite{cohen2019hierarchical} which is useful when the given weights denote dissimilarities instead of similarities. The formulation is the same as Dasgupta's, but now instead of minimization, the objective is phrased as a maximization\footnote{We restrict to binary trees as otherwise the problem would be trivial simply by outputting a root node with $n$ children.} problem:
\begin{equation}
  \max_{\mbox{binary tree } T} \sum_{(i, j) \in E} w_{i, j} |T_{i,j}| \label{eq:obj2}
\end{equation}
The current best approximation is given in~\cite{charikar2019hierarchical} and is a two-step algorithm based on Max-Cut~\cite{goemans1995improved} that finds a tree of value at least $0.667$ times that of the optimum tree as measured by (\ref{eq:obj2}). However, no hardness was known for this problem. Here, we show the following $\mathbf{APX}$-hardness result (even on graphs with $0$-$1$ edge weights) under Khot's Unique Games Conjecture (UGC)~\cite{khot2002power}:

\begin{theorem}
It is $\mathbf{UGC}$-hard to approximate Dissimilarity Hierarchical Clustering within a factor of $\tfrac{9159}{9189}\approx 0.9967$ (even on unweighted graphs).
\end{theorem}

\section{Inapproximability for Local Correlation Clustering}

In this section, we show that the $\ell_\infty$ version of Local Correlation Clustering (we simply refer to it as local-CC from now on) is $\mathbf{NP}$-hard to approximate within a $\tfrac43\approx 1.33$ factor even on complete graphs. 

\begin{theorem}\label{th:main}
Assuming $\mathbf{P}\neq \mathbf{NP}$, there is no polynomial time algorithm that can distinguish between the following two instances of Local Correlation Clustering (with $\ell_\infty$-norm) even on complete graphs:  
\begin{itemize}
    \item YES case: There is a clustering with value at most 3, i.e., at most 3 mistakes per vertex.

    \item NO case: Every clustering has value at least 4, i.e., at least 4 mistakes per vertex.
\end{itemize}
\end{theorem}
\noindent The proof of the theorem has several steps: We start from an instance of the $\mathbf{NP}$-hard {\sc Max 2-colorable degree 3-uniform hypergraph} problem and do a careful case analysis reducing it to the local-CC objective using a \textit{flower} and a \textit{bouquet} gadget.  A simpler version of our \textit{flower} gadget had previously been used in~\cite{charikar2005clustering}, however here we modify it by adding inner petal vertices (which we describe later); our \textit{bouquet} gadget is the same as in~\cite{charikar2005clustering}.

\paragraph{Coloring Hypergraphs:} Our reduction starts from {\sc Max 2-colorable bounded degree 3-uniform hypergraph} also used in~\cite{charikar2005clustering}: the input to this problem is a degree 3-uniform hypergraph $H = (V,S)$ where each hyperedge in $S = \{e_1,e_2, \cdots, e_m\}$ consists of exactly three elements of $V = \{v_1,v_2,\cdots, v_n\}$ with the added restriction that each element of $V$ occurs in at most $B$ hyperedges, for some absolute constant $B$ (so that $m \leq \frac{Bn}{3}$). The goal is to find a 2-coloring of $V$ that maximizes the number of hyperedges that are split by the coloring, i.e., they are bichromatic. It is known that for some absolute constants $\gamma > 0$ and $B$ (integer), given such a 3-uniform hypergraph, it is $\mathbf{NP}$-hard to distinguish between the YES case where $H$ is $2$-colorable (i.e., there exists a 2-coloring of the vertices under which no hyperedge is monochromatic) and the NO case where every 2-coloring of $V$ leaves at least $\gamma$ fraction of edges in $S$ monochromatic.

We shall prove that local-CC can capture an even harder version of the problem where there is no restriction on $B$ or $\gamma$. Formally, define the {\sc Max 2-colorable degree 3-uniform hypergraph} where the input is just a degree 3-uniform hypergraph $H = (V,S)$ as above, but with no parameters $B,\gamma$.

\begin{lemma}
It is $\mathbf{NP}$-hard to distinguish between the following two cases of {\sc Max 2-colorable degree 3-uniform hypergraph}:
\begin{itemize}
    \item YES case: $H$ is $2$-colorable.
    \item NO case: $H$ is not $2$-colorable, i.e., every $2$-coloring of $V$ leaves a hyperedge in $S$ monochromatic.
\end{itemize}
\end{lemma}
\begin{proof}
It is easy to see that any algorithm that could distinguish between the two cases of the {\sc Max 2-colorable degree 3-uniform hypergraph} problem, it would also distinguish between the two cases of the {\sc Max 2-colorable bounded degree 3-uniform hypergraph}, as the YES cases coincide, and every NO instance of the latter is a also a NO instance of the former.
\end{proof}

\paragraph{Constructing Flowers and Bouquets:} In our reduction, we construct a graph $G = (U,E)$ from the hypergraph instance $H = (V,S)$ using our \textit{flower} and  \textit{bouquet} gadgets. Firstly, for each vertex $v_i$ in the hypergraph, we construct a \textit{flower} structure $F_i$ with $6s_i$ vertices $U_i$, where $s_i$ is the number\footnote{If $s_i=1$, we can create a flower structure as if the vertex was in 2 hyperedges but without connecting  the extra outer petal vertices to anything. The same case analysis works in that case.} of hyperedges in which $v_i$ occurs (see also Figure~\ref{fig:flower}). The set $U_i$ consists of $2s_i$ vertices that form an induced cycle, and two pairs of $2s_i$ petal vertices each ($4s_i$ petals in total), that are adjacent to the two endpoints of the $2s_i$ cycle edges. Let $O_i$ ($E_i$) be the petal vertices drawn outside with odd (even) indices, according to an arbitrary cyclic ordering of the vertices as $1,2,\ldots, 2s_i$. Let $OI_i$ ($EI_i$) be the petal vertices drawn inside with odd (even) indices according to an arbitrary cyclic ordering of the vertices as $1,2,\ldots, 2s_i$. Secondly, consider a hyperedge   $e_j= (v_{j_1},v_{j_2},v_{j_3})$. To simplify presentation we set $j_1=1,j_2=2,j_3=3$ and ignore the subscript $j$. For our \textit{bouquet} gadget (see also Figure~\ref{fig:bouquet}), we create two independent edges $\alpha$ (with endpoints $A_1,A_2$) and $\beta$ (with endpoints $B_1,B_2$) in $G$. We add an edge from each endpoint $A_1,A_2$ to the vertex $O_{v_1}$ that corresponds to the occurrence of $v_{1}$ in $e$. Note that since $v_1$ participates in $s_{1}$ hyperedges, there are $s_{1}$ outside odd petals (e.g., $O_{v_1}$), and hence a different one can be used for each of the $s_{1}$ different hyperedges of $v_{1}$. The analogous edges are inserted between $A_1,A_2$ and the appropriate odd petals of the flowers corresponding to $v_2$ and $v_3$ (i.e., $O_{v_2}, O_{v_3}$). Finally, the endpoints  $B_1,B_2$ of the $\beta$ edge are similarly connected to the outside even petals $E_{v_1}$, $E_{v_2}$ and $E_{v_3}$. The instance for local-CC is then formed simply by labelling all edges of $G$ as positive and all non-existent edges as negative, thus obtaining the clique instance of the problem:

\begin{figure}
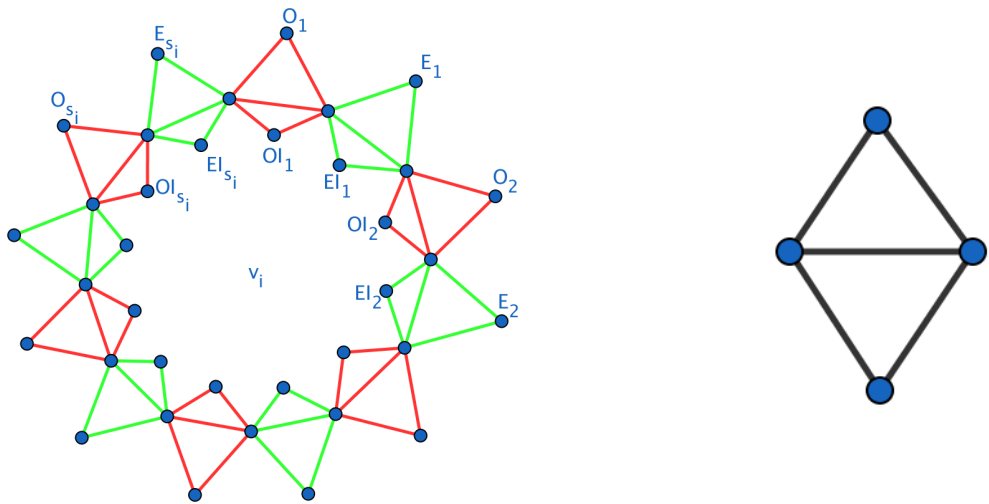

    \centering
    \begin{minipage}{0.5\textwidth}
        \centering
        \includegraphics[width=7cm]{flower.pdf}
        
    \end{minipage}\hfill
    \begin{minipage}{0.5\textwidth}
        \centering
        \includegraphics[width=3cm]{diamond.pdf}
        
    \end{minipage}
    \caption{(Left) The flower gadget of node $v_i$ that has $s_i$ hyperedges ($s_i=6$ in the figure). (Right) The diamond structure.}\label{fig:flower}\label{fig:diamond}
\end{figure}

\begin{figure}
\centering
  \includegraphics[width=12cm]{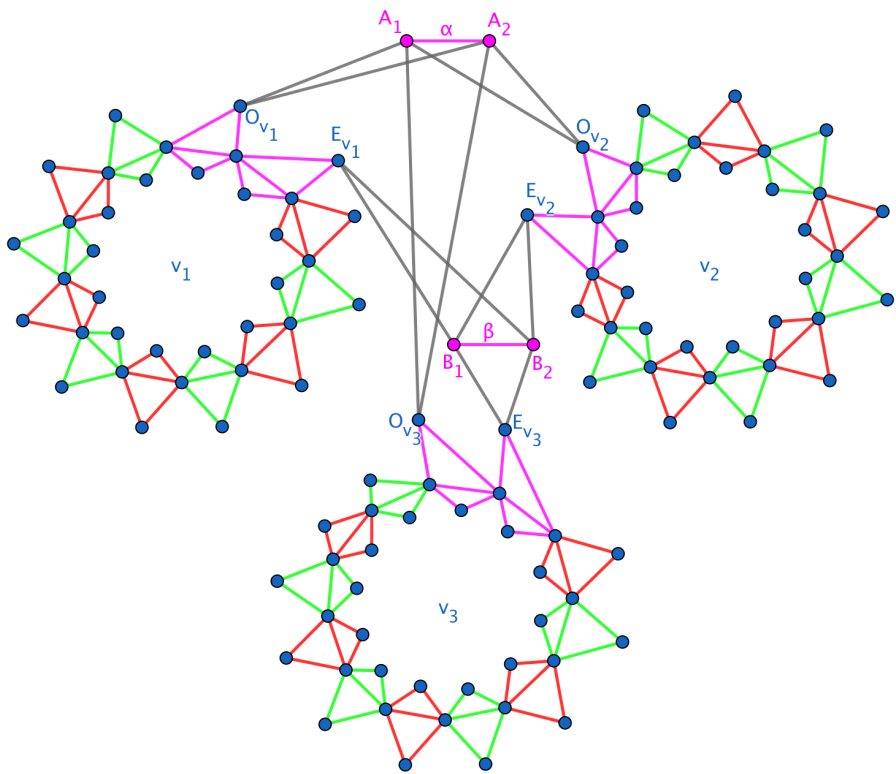}
  \caption{The bouquet gadget corresponding to a hyperedge on nodes $v_1,v_2,v_3$.}
  \label{fig:bouquet}
\end{figure}

\begin{lemma}\label{lem:main}
The hypergraph $H$ is 2-colorable if and only if the local-CC instance on $G$ has value at most $3$.
\end{lemma}

\noindent Lemma~\ref{lem:main} is the main technical component that allows us to connect the two problems and prove Theorem~\ref{th:main}. Before proceeding with the proof, we need to first derive a series of intermediate structural lemmas imposing constraints on any clustering of the graph $G$ having at most 3 mistakes per vertex. For the outer petals:

\begin{lemma}\label{lem:outer}
Let $C$ denote the cluster of any outer petal vertex (say $P3_1$). If $C$ does not contain any of $A_1,A_2,B_1,B_2$, then $C$ must be a diamond structure \footnote{The term \textit{diamond structure} denotes the clique on four nodes with a missing diagonal edge. For example, a diamond structure is formed by two vertices in the cycle edge together with its corresponding inner and outer petal as shown on right side of Figure~\ref{fig:diamond}.} containing the outer petal vertex ($P3_1$), its two polygonal neighbours in the flower structure ($Q3_1,Q3_2$) and the corresponding inner petal vertex ($R3_1$) to ensure at most $3$ mistakes.
%If an outer petal vertex (say $P3_1$) is not in the same cluster with any of $A_1,A_2,B_1,B_2$, then the only possible clustering (that ensures at most $3$ mistakes) for this petal, would be a diamond structure with its two neighbours in the polygon structure (say $Q3_1,Q3_2$) and the corresponding inner petal vertex ($R3_1$).
\end{lemma}

\begin{proof}
If cluster $C$ does not contain any of $A_1,A_2$ (see Figure~\ref{fig:detailed}), then vertex $P3_1$ already has two mistakes (corresponding to $A_1,A_2$). Thus, to ensure at most $3$ mistakes for vertex $P3_1$, cluster $C$ needs to have at least one of $Q3_1$ or $Q3_2$ (or both) in it since they are both neighbors of $P3_1$. 

Consider the case when we add only one of $Q3_1$ or $Q3_2$ to cluster $C$. Wlog, we can try to include $Q3_2$ in $C$ but not $Q3_1$. Now, even if we include $R3_1$ in this cluster, $P3_1$ has $3$ mistakes (corresponding to $A_1,A_2,Q3_1$) and $Q3_2$ has $4$ mistakes (corresponding to $ Q3_1, P3_2, Q3_3, R3_2$) and we cannot include any other neighbour of $Q3_2$ (except $Q3_1$ which is excluded by assumption) in the cluster $C$ since it would lead to an increase in mistakes for $P3_1$ (as they are non-neighbours of $P3_1$). Hence, a cluster like this is excluded. We can similarly argue for the case when we include $Q3_1$ in cluster $C$ but not $Q3_2$. 

The only remaining case is to include both $Q3_1$ and $Q3_2$ in the cluster $C$. If we do not include $R3_1$ in this cluster $C$, both $Q3_1$ and $Q3_2$ have different disjoint neighbours which are also not neighbours of $P3_1$; trying to include 1 neighbour for each of them, would lead to $P3_1$ having 4 mistakes. Hence, we would have to include $R3_1$ in the cluster $C$, which proves the lemma as no other vertex can be included and still obtain at most 3 mistakes. In other words, if an outer petal vertex does not have neighbouring vertices corresponding to $\alpha$ or $\beta$ edges in its cluster, the only possible option for it is to form a diamond with the corresponding polygonal vertices and the inner petal vertex.
\end{proof}

\begin{lemma}\label{lem:singleton}
Let $C$ be the cluster containing any outer petal vertex (say $P1_1$). If $C$ contains at least one of $A_1,A_2,B_1,B_2$, then 
\begin{itemize}
    \item[a)] $C$ cannot contain any other vertex from the flower structure corresponding to $P1_1$.
    \item[b)] The corresponding inner petal vertex $R1_1$ would form a singleton cluster.
    \item[c)] The corresponding polygonal vertices in the flower structure would form a diamond cluster, that is, there would be a cluster containing $Q1_2, Q1_3, P1_2, R1_2$ and similarly on the other side.
\end{itemize}

%If an outer petal vertex, say $P1_1$, has at least one of $A_1,A_2,B_1,B_2$ in its cluster $C$, then $C$ cannot contain any other vertex from the flower structure. Moreover, the only possible clustering after this would be to form a singleton cluster with the corresponding inner petal vertex, and diamond clusters with the neighbouring inner and outer petals, say $Q1_2, P1_2, Q1_3, R1_2$ (and similarly on the other side).
\end{lemma}

\begin{proof}
We start with the first part of the claim (we focus on $A_1$ due to symmetry). All possible cases are:

\begin{itemize}
    \item $C$ contains $A_1, A_2, P1_1, P2_1, P3_1$: In this case, each of  $P1_1, P2_1, P3_1$ has 4 mistakes and hence, one of its neighbours has to be in the cluster $C$ which would increase the mistakes of $A_1, A_2$.
    \item $C$ contains $A_1, A_2, P1_1, P2_1$ but not $P3_1$: If we bring one neighbour among $Q1_2$ or $Q1_1$ (say $Q1_1$) into the cluster, then  since $Q1_1$ has 5 other positive neighbours, 2 of them would have to be brought into the cluster as well, leading to an increase in the mistakes of all other non-neighboring vertices. The other case where $C$ contains $A_1,A_2,P1_1,P3_1$ would be symmetric. 
    \item $C$ contains $A_1, A_2, P1_1$ but not $P2_1, P3_1$: The same argument as in the previous case holds.
    \item $C$ contains $A_1, P1_1$: The same argument as in the previous case holds.
    \item $C$ contains $A_1, P1_1, P2_1$: The same argument as in the previous case holds. The case where cluster $C$ contains $A_1, P1_1, P3_1$ is symmetric.
    \item $C$ contains $A_1, P1_1, P2_1, P3_1$ but not $A_2$: This cluster is not possible since $A_2$ has 4 mistakes.
\end{itemize}

For the second part of the claim, we need to argue that if $P1_1$ is together with $A_1$ or $A_2$, then $R1_1$ would form a singleton cluster. We already know that $Q1_1, Q1_2$ cannot go to cluster $C$ containing $P1_1$. Let us say we put $Q1_1, Q1_2$ into the same cluster. Now since both $Q1_1, Q1_2$ have 5 remaining positive neighbours (and hence 5 mistakes), we need to put at least 2 more into the cluster for each of them. The only plausible case is putting $R1_1$ as it is a neighbour of both. Now since the other neighbours of $Q_1, Q_2$ are disjoint, if we try to put one neighbour of each into the cluster, the mistakes would not decrease, hence, this clustering (where $Q1_1, Q1_2$ are put into the same cluster) is not possible. Hence, we know that $Q1_1, Q1_2$ do not belong to the same cluster.

Now, let us say $R1_1$ is not in a singleton cluster. Assume wlog it is with $Q1_2$. As $Q1_2$ already has 2 mistakes corresponding to edges $Q1_1$ and $P1_1$, we need to put 2 additional of its neighbours into the cluster. If we put $Q1_3$ into the cluster, we would also need to put $Q1_3$'s neighbours into the cluster leading to more than 4 mistakes per vertex for $R1_1$. Instead of $Q1_3$, if we put $R1_2$ and $P1_2$, then $P1_2$ has 5 mistakes corresponding to $B_1,B_2,Q1_3,R1_1, R1_2$ and further adding any of its neighbours to the cluster leads to $R1_1$ having 4 mistakes. Hence, $R1_1$ has to be in a singleton cluster. This completes the proof of the second part of the lemma.

For the third part, we need to show that $Q1_2, P1_2, Q1_3, R1_2$ would form a cluster and similarly on the other side. Now, $Q1_2$ already has 3 mistakes corresponding to $P1_1, Q1_1, R1_1$ and hence all its other 3 neighbours need to be added to its cluster, i.e. $P1_2, R1_2, Q1_3$.  Since in this cluster $Q1_2$ already has 3 mistakes corresponding to its 3 neighbours, no other vertex can be added. Hence, the only possibility is the diamond cluster that was claimed. An identical argument can be made for the other side of the polygon.
\end{proof}

 \begin{figure}
    \centering
  \includegraphics[width=10cm]{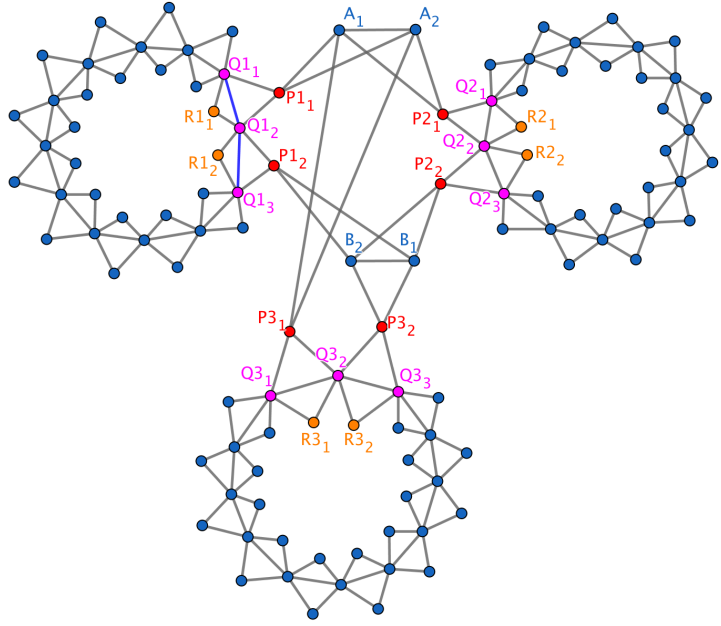}
  \caption{Figure of a bouquet used in the proof of many structural lemmas.}
  \label{fig:detailed}
\end{figure}
\begin{figure}
\centering
  \includegraphics[width=10cm]{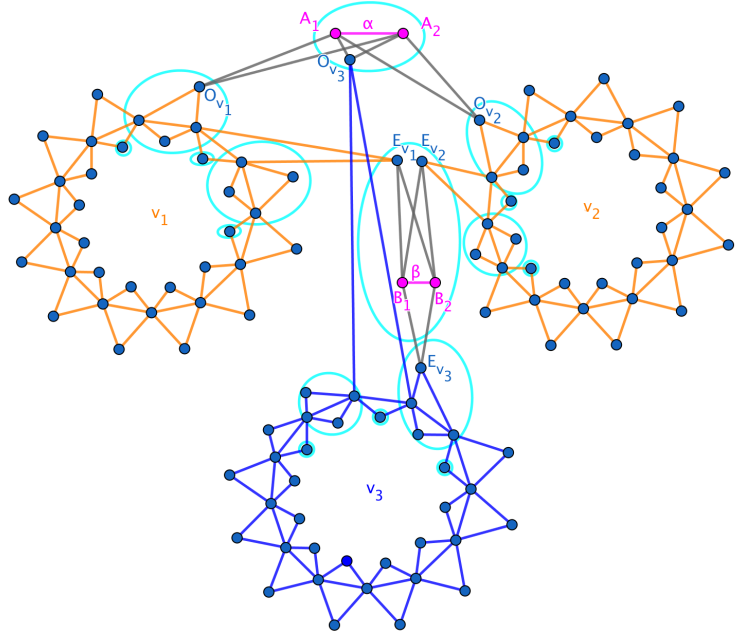}
  \caption{YES case and the clustering with at most 3 mistakes at each node.}
  \label{fig:yes}
\end{figure}

\begin{lemma}\label{lem:outerpetalvertex}
Let $C$ denote the cluster of an outer petal vertex (say $P1_1$). Then, $C$ must be of the form of the following four options to ensure at most $3$ mistakes per vertex. 
\begin{itemize}
    \item[a)] $C$ contains $P1_1, A_1$ (or symmetrically, $C$ contains $P1_1, A_2$).
    \item[b)] $C$ contains $P1_1, A_1, A_2$.
    \item[c)] $C$ contains $P1_1, P2_1,A_1, A_2$ (or symmetrically, $C$ contains $P1_1, P3_1,A_1, A_2$).
    \item[d)] $C$ is a diamond structure containing $P1_1, Q1_1, Q1_2, R1_1$.
\end{itemize}
\end{lemma}
\begin{proof}
First of all, note that $P1_1$ has four positive edges, thus it cannot form a singleton cluster. It has to include at least one of its neighbours in its cluster. We consider the following four cases:

\begin{itemize}
    \item $C$ includes $P1_1, A_1$ but does not include $A_2$: From Lemma~\ref{lem:singleton}, we know that the cluster $C$ cannot include any other vertex from the flower structure (like $Q1_1, Q1_2, R1_1$). In cluster $C$, $P1_1$ already has three mistakes and thus, we cannot add any other vertex to this cluster. Symmetrically, another option is cluster $C$ contains $P1_1, A_2$.
    \item $C$ includes $P1_1, A_1, A_2$ but does not include $P2_1, P3_1$: From Lemma~\ref{lem:singleton}, we know that the cluster $C$ cannot include any other vertex from the flower structure (like $Q1_1, Q1_2, R1_1$). The other possible vertices that we could add to the cluster $C$ include $B_1$ or $B_2$, but since $B_1$ and $B_2$ both have 4 positive neighbours each, we will also have to include their neighbours which would drive the mistakes of $P1_1$ up to 4. 
    \item $C$ includes $P2_1, P1_1, A_1, A_2$: From Lemma~\ref{lem:singleton}, we know that cluster $C$ cannot include any other vertex from the flower structure (like $Q1_1, Q1_2, R1_1$). We cannot include any other vertex in $C$ since $P1_1$ already has 3 mistakes in this cluster. Symmetrically, another option is cluster $C$ contains $P1_1, P3_1, A_1, A_2$.
    \item $C$ does not include $A_1, A_2$. From Lemma~\ref{lem:outer}, we know that the only possible option for the cluster $C$ is the diamond structure including $P1_1, Q1_1, Q1_2, R1_1$. 
\end{itemize}
\end{proof}

\begin{lemma}\label{lem:diamond}
In order to have at most 3 mistakes per vertex, for every flower structure, either all the odd vertices form diamond clusters or all the even vertices form diamond clusters.
\end{lemma}
\begin{proof}
From Lemma~\ref{lem:outerpetalvertex}, we know that either an outer petal vertex forms a diamond structure or it goes with endpoints of the $\alpha$ or $\beta$ edges. Note that it cannot happen that no outer petal vertex in a flower structure forms a diamond structure, because if an outer petal vertex goes with $\alpha$ or $\beta$ edges, then the neighbouring outer petal vertex would have to form a diamond structure (Lemma~\ref{lem:singleton}). Therefore, if one of the odd outer petal vertices $P1_1$ forms a diamond cluster, the neighbouring even outer petal vertices ($P1_2$ and $P1_{s_1}$) would go with the $\beta$ edges. Now, using Lemma~\ref{lem:singleton} and Lemma~\ref{lem:outerpetalvertex}, the odd outer petal vertices ($P1_3$ and $P1_{s_1-1}$) would again form diamond clusters. Continuing in this manner, we get that either all the odd outer vertices form diamond clusters or the even outer petal vertices form diamond clusters.  
\end{proof}

\begin{lemma}\label{lem:notallthree}
In order to have at most 3 mistakes per vertex, for every hyperedge, not all three flower structures corresponding to it can have diamond structures corresponding to even vertices (or odd vertices).
\end{lemma}

\begin{proof}
Consider that for a hyperedge as in Figure~\ref{fig:detailed}, all its three flowers have diamonds corresponding to odd vertices. Therefore, vertices $P1_2$, $P2_2$ and $P3_2$ all have to include at least one vertex among $B_1, B_2$ (Lemma~\ref{lem:outerpetalvertex}). Based on the allowed clusterings for outer petals in Lemma~\ref{lem:outerpetalvertex}, this leads to 4 mistakes.
\end{proof}

\begin{proof}[Proof of Lemma~\ref{lem:main}]
\textbf{YES case:}
In this case, $H$ is 2-colorable and we show how to construct a clustering of $G$ with at most $3$ mistakes per vertex. Let $f:V \to \{Orange, Blue\}$ be the coloring function mapping vertices to colors, such that every hyperedge of $H$ is bichromatic (see also Figure~\ref{fig:yes}). For each flower structure $F_i$, if the corresponding vertex is $Orange$, we pick the $s_i$ diamond structures of the odd collection, otherwise if the vertex is $Blue$, we pick the $s_i$ diamond structures belonging to the even set. If the vertex is $Orange$, we assign all inner petal vertices of the even collection of the flower as singleton clusters, otherwise if the vertex is $Blue$, we assign all inner petal vertices of the odd collection of the flower as singleton clusters.  Since we know that each hyperedge $e_j$ is bichromatic, we assume that two of its vertices $v_{1}, v_{2}$ are colored $Orange$ and the third vertex $v_{3}$ is colored $Blue$.  Then for each hyperedge $e_j$, we can choose two clusters as follows: one is the triangle containing the edge $\alpha_j$ together with its neighbour in $O_{v_3}$ and the other is the diamond containing the edge $\beta_j$ together with its neighbours in $E_{v_1}$ and $E_{v_2}$. For this clustering, each vertex has at most $3$ mistakes, as desired. The case when two of the vertices are colored $Blue$ and the third one $Orange$ would be symmetric.

\textbf{NO case:}
If the initial hypergraph $H$ is not 2-colorable, we show that every possible clustering of the vertices in $G$ incurs at least 4 mistakes on some vertex. In fact, we show the contrapositive: if there exists a clustering that has at most 3 mistakes for every vertex in $G$, then it is possible to construct a 2-coloring of the hypergraph $H$ with all hyperedges being bichromatic. From Lemma~{\ref{lem:diamond}}, we know that every flower structure either has all its odd vertices in the diamond structures, or it has all its even vertices in the diamond structures. Thus, we can color the corresponding vertex $Orange$ if the odd vertices form diamond structures and $Blue$ if the even vertices form diamond structures. Moreover, from Lemma~{\ref{lem:notallthree}}, we know that not all three flowers corresponding to a single hyperedge can have clusters containing all their odd diamonds ($Orange$ vertices) or all their even diamonds ($Blue$ vertices). Hence, if there exists a clustering with at most 3 mistakes per vertex, we can create a bichromatic coloring for every hyperedge and this completes the proof of the main lemma. 
\end{proof}

\begin{proof}[Proof of Theorem~\ref{th:main}]
Now that we have Lemma~\ref{lem:main} it is easy to see that local-CC captures the {\sc Max 2-colorable degree 3-uniform hypergraph} problem. Starting from the latter problem, if a given hypergraph was a YES instance, then the constructed graph by our reduction would have a local-CC solution incurring only 3 mistakes per vertex, otherwise, if it was a NO instance, all clusterings would incur 4 mistakes or more. 
\end{proof}

\section{Inapproximability for Dissimilarity Hierarchical Clustering}

Given a weighted graph $G = (V, E, w)$, let us consider objective~(\ref{eq:obj2}). For simplicity, we can assume that $\sum_{e} w_e = 1$ and that the objective is divided by $n$, so that the optimal value is always at most $1$. We use the following standard definition for correlated Gaussians:
\begin{definition}
Let $\Phi$ be the cumulative distribution function of a standard Gaussian variable (e.g., $\Pr[x \leq \Phi^{-1}(a)] = a$). For $\rho \in [-1, 0]$ and $a, b \in [0, 1]$, we define
\[
\Gamma_{\rho}(a, b) := \Pr[x \leq \Phi^{-1}(a) \wedge y \leq \Phi^{-1}(b)], 
\]
where $(x, y)$ are correlated Gaussians with the covariance matrix $((1, \rho), (\rho, 1))$.
\end{definition}

\begin{theorem} 
It is $\mathbf{UGC}$-hard to approximate the Dissimilarity HC objective~(\ref{eq:obj2}) within a $0.9967$ factor.
\end{theorem}
\iffalse
Fix $\rho \in [-1, 0]$ and let $\alpha = (1 - \rho) / 2$. Assuming the UGC, given an weighted graph $G = (V, E)$, it is NP-hard to distinguish the following cases.
\begin{itemize}
\item YES case: There exists a hierarchical clustering with cost $\frac{\alpha}{1 - (1- \alpha)/2}$.
\item NO case: Any subset $S \subseteq V$ with $|S| = \beta |V|$ induces at least $\Gamma_{\rho}(\beta, \beta)$ fraction of edges. 
\end{itemize}
\fi

\begin{proof}
The starting point of our proof is the hardness result of \cut and \lin by Khot et al.~\cite{khot2007optimal}. 
The \lin problem is defined as follows:
\begin{itemize}
    \item Variables: $X = \{x_1, \dots, x_{n_0} \}$ where each $x_i$ can take a value from $\Z_q$. 
    \item Input: A set of equations where the $j$th equation is of the form $x_{j,1} - x_{j,2} = a_j$ for $x_{j,1}, x_{j,2} \in X$ and $a_j \in \Z_q$. 
    \item Goal: Find an assignment $\sigma : X \to \Z_q$ that maximizes the number of satisfied equations. Let the value of an instance to be the maximum fraction of equations satisfied by any assignment. 
\end{itemize}

Khot et al. showed that the hardness of \lin is equivalent to the original UGC. More precisely, assuming UGC, for any $\eps > 0$, there exists $q \in \N$ such that it is $\mathbf{NP}$-hard to distinguish whether a given instance of \lin has value at least $1 - \eps$ or at most $\eps$. 
In the YES case, there exists an assignment $\sigma : X \to \Z_q$ satisfying at least an $1-\eps$ fraction of constraints. We crucially use the fact that any {\em shift} of $\sigma$ also achieves the same value; for any $a \in \Z_q$, the shifted assignment $\sigma_a(x) := \sigma(x)+a$ satisfies exactly the same constraints as $\sigma$. 
Furthermore the result holds even when the \lin instance is {\em regular} --- each variable in $X$ is contained in the same number of equations. 

The same paper~\cite{khot2007optimal} proved a hardness for \cut (under UGC), by giving a reduction from Unique Games to \cut. We apply this reduction, but starting from the aforementioned hard instance of \lin instead of an arbitrary Unique Games instance. The resulting instance for \cut is our hard instance for HC. Our reduction is parameterized by $\rho$ (we later set $\rho = -0.7$). It produces a weighted graph whose vertex set is $V = X \times  \{ \pm 1 \}^{\Z_q}$, where $X$ is the set of variables for \lin. 
For each pair $(u, v) \in \binom{V}{2}$, the weight of $(u, v)$ to be the defined to be the probability that it is sampled in the following following procedure. 

\begin{itemize}
    \item Sample a random variable $x_i$. Then, sample two random constraints involving $x_i$ (say $x_j - x_i = a, x_k - x_i = b$). Next, sample $f, g \in \{\pm1\}^{\Z_q}$ such that for every $r \in \Z_q$ independently, $f_{r+a}$ and $g_{r+b}$ are sampled from the $\rho$-correlated, mean-zero $\{ \pm 1 \}$ distribution (two values are the same with probability $\tfrac{1 + \rho}{2}$). Finally, output the pair $\{ (x_j, f), (x_k, g)\}$.
\end{itemize}

Let $n$ be the number of vertices in the final instance. In the YES case, there exists an assignment $\sigma : X \to \Z_q$ that satisfies at least an $(1 - \eps)$ fraction of constraints of \lin for an arbitrary constant $\eps > 0$ we can choose. Since the instance for \lin is regular, sampling a random $x_i$ and a random constraint $x_j - x_i = b$ involving $x_i$ is the same as sampling a uniformly random equation. By union bound, the probability that $\sigma$ satisfies both is at least $1 - 2\eps$. 

Consider the above sampling procedure given 
$x_i$ and two constraints $x_j - x_i = a$, and $x_k - x_i = b$.
When $\sigma$ satisfies both constraints, 
the probability that $f_{\sigma(x_j)} = g_{\sigma(x_k)}$ is exactly $(1 + \rho) / 2$. 
Since $\sigma_c := \sigma + c$ for any $c \in \Z_q$ satisfies the same set of constraints, the same statement holds for $\sigma_c$ as well. Moreover, the coordinates of $f$ and $g$ are sampled independently given $r$, for any $C \subseteq \Z_q$,
\begin{equation}
\Pr \bigg[ f_{\sigma_c(x_j)} = g_{\sigma_c(x_k)} \mbox{ for all } c \in C \bigg] = \bigg( \frac{1+\rho}{2} \bigg)^{|C|}
\label{eq:ind}
\end{equation}

Let $S_{c, 1} := \{ (x_i, f) : f_{\sigma_c(x_i)} = 1 \}$ and 
$S_{c, -1} := \{ (x_i, f) : f_{\sigma_c(x_i)} = -1 \}$. For every $c \in \Z_q$, $(S_{c, 1}, S_{c, -1})$ is a partition of $V$. Let $T$ be the tree such that the root is at level 0 and for each $r \in [1, q]$, each node in the $r$th level is labeled by $b = \{ \pm 1 \}^r$ and corresponds to $\cap_{i=1}^r S_{i, b_i}$. (We use the natural mapping between $[1, q]$ and $\Z_q$ with $q = 0$.) 
For $1 \leq r < q$, a node in the $r$th level becomes the parent of a node in the $(r+1)$th level if the label of the former is a prefix of the label of the latter. The remaining levels of $T$ can be constructed arbitrarily. This gives a perfectly balanced binary tree up to $q$th level. Due to~\eqref{eq:ind}, the total weight of edges split in the $r$th level is at least $(1 - 2\eps) \cdot ((1+\rho)/2)^{r-1} \cdot ((1-\rho)/2)$. Let $\alpha = (1 - \rho)/2 = 0.85$. In the objective function for HC, these edges are multiplied by the number of vertices corresponding to a node in the $(r-1)$th level, which is $n \cdot 2^{1-r}$. Therefore, the objective value for $T$ is at least
\begin{align*}
\sum_{r = 1}^{q} 
\bigg(
(1-2\eps)(1-\alpha)^{r-1} \alpha \cdot
n 2^{1-r} \bigg)
=
n (1 - 2\eps) \sum_{r = 1}^{q} \alpha \bigg(\frac{1 - \alpha}{2}\bigg)^{r - 1}= n (1 - 2\eps) \bigg( \frac{\alpha}{1 - (1- \alpha)/2} - 2^{\Omega(-q)}\bigg).
\end{align*}
We set $\rho = -0.7$, which is close to $\mathrm{argmin}_\rho \tfrac{\arccos \rho / \pi}{(1-\rho)/2} \approx -0.689$ used for the optimal $0.878$-hardness of \cut.
This yields $\alpha = 0.85$, and make the objective function value at least $0.9189n$ for small enough $\eps > 0$ and large enough $q \in \N$.

Now we analyze the objective value for HC in the NO case. In the NO case, Khot et al.~\cite{khot2007optimal} showed the following statement holds for any small constant $\eps > 0$. 
\[
    \mbox{Any subset } S \subseteq V \mbox{ with }|S| = \beta |V| \mbox{ induces edges of weight at least } \Gamma_{\rho}(\beta, \beta) - \eps.
\label{eq:kkmo}
\tag{*}
\]

For simplicity, we ignore $\eps$ in the later calculations as all of them hold for small enough $\eps$. 

Let $T$ be any tree. For a node $v$, let $\beta_v$ be the number of leaves in the subtree rooted at $v$ divided by $n$. 
 Then~\eqref{eq:kkmo} implies that edges of weight at least $\Gamma_{\rho}[\beta_v, \beta_v]$ will be multiplied by $\beta_v$ in the objective function. Even when all the other edges (of total weight at most (1 - $\Gamma_{\rho}[\beta_v, \beta_v]$)) are multiplied by $n$, the total objective function value for $T$ divided by $n$ is at most 
\begin{equation}
(1 - \Gamma_{\rho}[\beta_v, \beta_v]) + \beta_v \cdot \Gamma_{\rho}[\beta_v, \beta_v].
\label{eq:plot1}
\end{equation}

To upper bound the value of $T$, we consider the following two scenarios and show that the value is small in both cases:
\begin{itemize}
\item
If there exists $v$ such that $\beta_v \in [0.6, 0.88]$, we can confirm from the left of Figure~\ref{fig:plot1} that the value of~\eqref{eq:plot1} is at most $0.909$ --- it is maximized when $\beta_v = 0.88$ and all other $\beta_v$ in the interval achieves a smaller value. 

%When $\beta_v = 0.88$, this value is $0.909$ which is the maximum over all $\beta_v \in [0.6, 0.88]$ (see Figure~\ref{fig:plot1}).

\begin{figure}
    \centering
    \begin{minipage}{0.5\textwidth}
        \centering
        \includegraphics[width=0.75\textwidth]{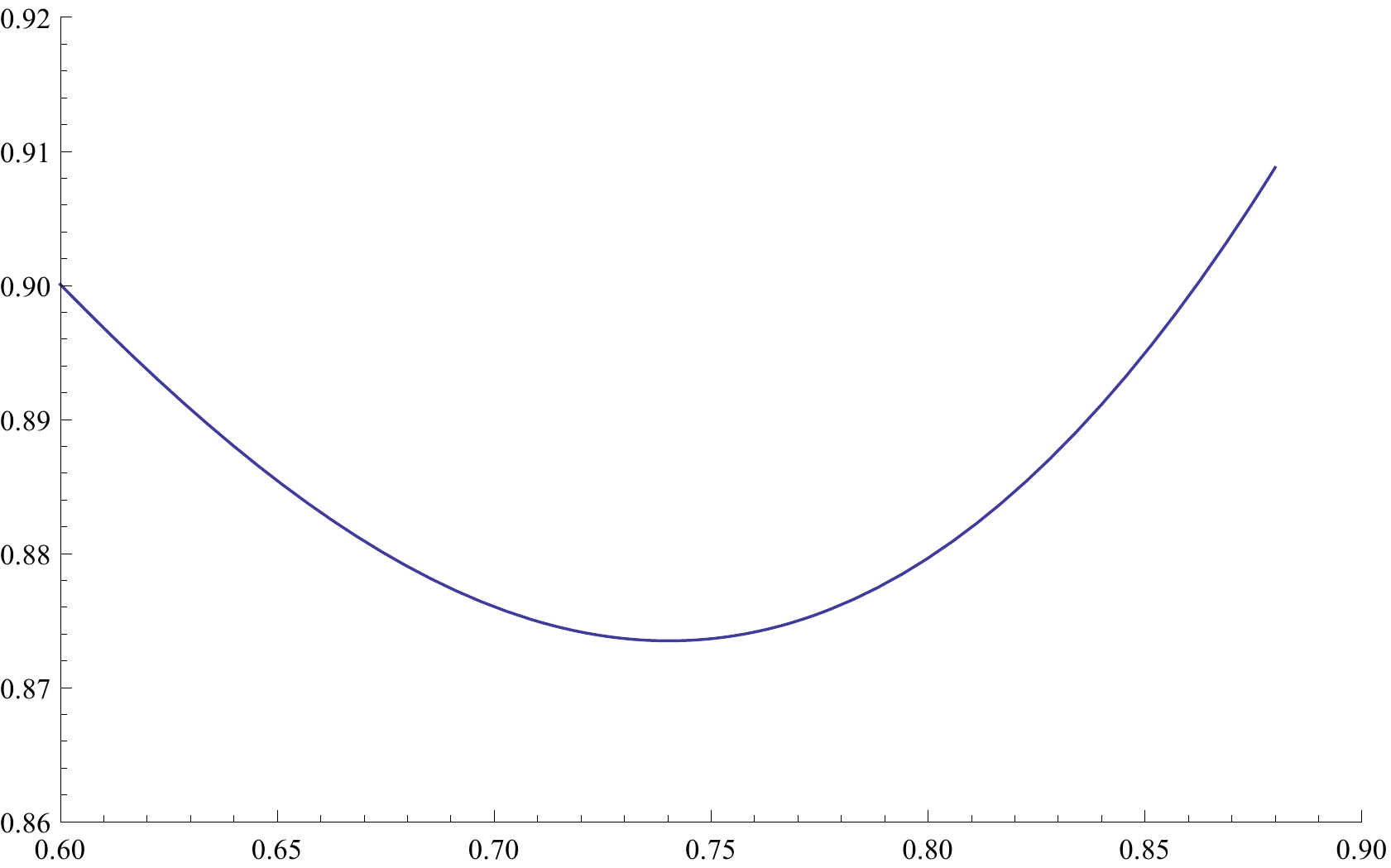}
       
    \end{minipage}\hfill
    \begin{minipage}{0.5\textwidth}
        \centering
        \includegraphics[width=0.75\textwidth]{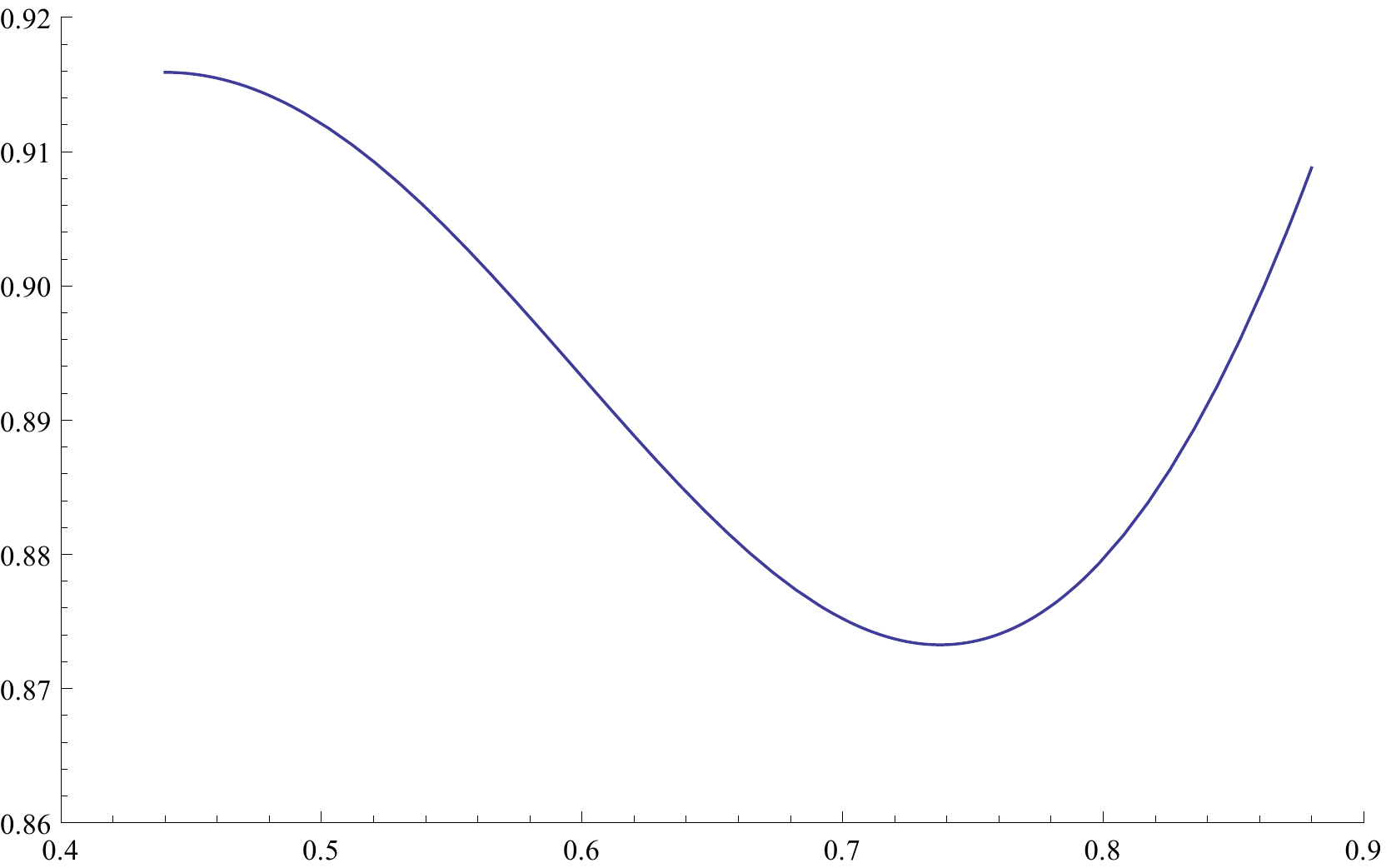}
        
    \end{minipage}
    \caption{(Left) Plot of~\eqref{eq:plot1} for $\beta_v \in [0.6, 0.88]$. (Right) Plot of~\eqref{eq:plot2} for $\beta_{u_2} = 0.88 - \beta_{u_1}$, $\beta_{u_1} \in [0.44, 0.88]$.} \label{fig:plot1}\label{fig:plot2}
\end{figure}

\item Suppose that no such $v$ with $\beta_v \in [0.6, 0.88]$ exists. Let $u$ be the node with the smallest $\beta_u$ among vertices with $\beta_u > 0.88$. Let $u_1, u_2$ be the children of $u$. By definition, $\beta_{u_1}, \beta_{u_2} < 0.6$, which means that both $\beta_{u_1}, \beta_{u_2} > 0.28$.  Letting $\Gamma_i := \Gamma_{\rho}[\beta_{u_i}, \beta_{u_i}]$ for $i = 1, 2$, again by~\eqref{eq:kkmo}, 
the edges of the HC instance induced by the vertices in the subtree of $u_i$ have total weight at least $\Gamma_i$ and will be multiplied by  $\beta_{u_i} n$. Similarly to~\eqref{eq:plot1}, even when all the other edges are multiplied by $n$, the objective function value divided by $n$ has value at most:
\begin{equation}
(1 - \Gamma_1 - \Gamma_2) + \beta_{u_1} \Gamma_1 + \beta_{u_2} \Gamma_2.
\label{eq:plot2}
\end{equation}
Since $\Gamma_{\rho}(\beta)$ is monotonically increasing in $\beta$, the above expression is monotonically decreasing in $\beta_{u_1} + \beta_{u_2}$, so it is maximized when $\beta_u = \beta_{u_1} + \beta_{u_2} = 0.88$. 
Without loss of generality assume $\beta_{u_1} \geq \beta_{u_2}$, so that $\beta_{u_1} \in [0.44, 0.88]$.
From the right of Figure~\ref{fig:plot2}, 
when $\beta_{u_1} \in [0.44, 0.88]$, the value of~\eqref{eq:plot2} is at most $0.9159$, maximized when $\beta = 0.44$. 
\end{itemize}

Therefore, in both scenarios the objective function value is at most $0.9159n$, strictly less than the YES case where the value is $0.9189n$. This proves the desired hardness to within a $\tfrac{9159}{9189}$ factor.
\end{proof}

\iffalse

\begin{remark}
For the \cut hardness, in the NO case, for any bisection $(S, T)$, the above theorem shows that 
each part induces at least $\Gamma_{-0.7}(0.5, 0.5) \geq 0.1265$ fraction of edges. 
(And some more calcuations will show that bisection is the best you can do.)
Therefore the total weight of edges cut is at most $1 - 2 \cdot 0.1265 = 0.747$. Since we cut $\alpha = 0.85$ fraction of edges in the YES case, 
the hardness ratio of \cut we get is $0.767 / 0.85 \approx 0.878$, which is the Goemans-Williamson ratio.  
\end{remark}

\fi

\bibliographystyle{alpha}
\bibliography{main}
\end{document}